\documentclass[12pt,fleqn,leqno]{article}
\usepackage{amsmath,amssymb}
\usepackage{amsthm}
\usepackage[]{graphicx}

\allowdisplaybreaks

\numberwithin{equation}{section}

\makeatletter
\renewcommand{\section}{\@startsection{section}{1}{0pt}{20pt}{6pt}{\large\bf}}
\renewcommand{\@seccntformat}[1]{\csname the#1\endcsname.\ }

\def\footnoterule{\kern -3pt \hrule width 2.7 true cm \kern 2.6pt}

\def\ni{\noindent}
\def\vs{\vspace}
\def\hs{\hspace}

\def\EE{\mathbb E}

\def\QQ{\mathbb Q}

\def\R{I\!\!R}
\def\L{I\!\!L}

\def\wh{\widehat}

\newcommand{\p}{\! +\! }
\newcommand{\m}{\! -\! }

\newtheorem{theorem}{Theorem}[section]

\newtheorem{definition}[theorem]{Definition}

\begin{document}

\title{\textbf{On the valuation of multiple reset options: integral equation approach}}
\author{Nazym Azimbayev\thanks{Department of Computational Mathematics and Cybernetics, Lomonosov Moscow State University, Moscow, Russia; \texttt{nazym.azimbayev@cs.msu.ru}} \thanks{The National Bank of Kazakhstan, Nur-Sultan, Kazakhstan}\quad \& Yerkin Kitapbayev\thanks{Department of Mathematics, NC State University, Raleigh NC, USA; \texttt{ykitapb@ncsu.edu}}}
\maketitle


{\par \leftskip=2.6cm \rightskip=2.6cm \footnotesize

 In this paper, we study a pricing problem of the multiple reset put option, which allows the holder to reset several times a current strike price to obtain an at-the-money European put option. 
We formulate the pricing problem as a multiple optimal stopping problem, then reduce it to a sequence of single optimal stopping problems and study the associated free-boundary problems.
We solve this sequence of problems by induction in the number of remaining reset rights and exploit probabilistic arguments such as local time-space calculus on curves. As a result, we characterize each optimal reset boundary as the unique solution to a nonlinear integral equation and derive the reset premium representations for the option prices. 
We propose that the multiple reset options can be used as cryptocurrency derivatives and an attractive alternative to standard European options due to the extreme volatility of underlying cryptocurrencies.

\par}





\vs{-18pt}

\vs{-18pt}

\section{Introduction}


Let us imagine an investor who holds a standard European call or put option with strike price $K$ and maturity date $T$. Then there are at least two possible scenarios when the holder can feel a regret: 1) there is a period before $T$ when underlying asset price movements are favorable for him; however, he cannot early exercise his option, and then the asset price will turn back and at time $T$ he gets small or zero payoff; 2) the option becomes deep out-of-the-money, and the asset price is below $K$ for the call (above $K$ for the put) option, and most likely he gets zero payoff at time $T$. In the last two decades, the options with shouting and reset features have been introduced and studied, addressing these unfavorable scenarios. They can be divided into the two groups of options: 1) shout (call or put) option, which allows the holder to lock the profit at some endogenous time $\tau$  and then at time $T$ get the maximum between two intrinsic values at $\tau$ and $T$; 2) reset (call or put) option gives to the investor the right to reset the strike $K$ to the current price, i.e., to substitute the current out-of-the-money option to the at-the-money one. The first group, i.e., shout options, allows the investor to lock the profit while having the opportunity of payoff increase at $T$. In contrast, the second one increases the chances to get a positive payoff at $T$. The pricing problem for both options can be formulated as optimal stopping problems where stopping times represent shouting or reset times, respectively. They have both European (since the payoff is known at $T$ only) and American features (due to early `shouting' or `reset' opportunities).

Below we provide a literature review on the reset and shout options. These options might have both single and multiple opportunities to shout or reset. However, so far, the research has been focused primarily on single shout and reset options.
The origin of the shout option goes to the paper \cite{Thomas} and brief analysis can be found in textbooks \cite{Hull} and \cite{Wilmott} where the binomial tree method is offered for pricing the option. There are several papers where these options were thoroughly studied from both theoretical and numerical points of view. In series of works \cite{WFV-1}-\cite{WFV-3} several sophisticated numerical schemes have been developed to price the reset and shout options. The comprehensive analysis has been performed in \cite{Kwok-1}, where the reset put option was formulated as a free-boundary problem,  shouting premium for the option price and recursive integral equation for the optimal shouting boundary were obtained. In \cite{Kwok-2}, authors studied the multiple reset put option problem, i.e., the buyer has several rights to reset a strike. 
The numerical experiments and analysis have been provided using the binomial method approach. Several monotonicity properties for option prices and optimal reset boundaries were established through theoretical arguments and numerical experiments.

Later, a thorough theoretical analysis was provided in \cite{Dai-1}, where PDE and variational inequality approaches were applied to show the existence and uniqueness of a solution to the free-boundary problem associated with the reset put option problem. Also, monotonicity and regularity of the optimal shouting boundary have been shown in some cases. Then in \cite{Alobaidi-1},
using a Laplace transform, the Fredholm integro-differential equations for optimal shouting boundaries of shout call and put options were obtained. Finally, in \cite{Goard} the formal series expansions have been discovered for the option price and optimal shouting boundaries of the reset put and shout call options.

This paper is devoted to the pricing of the multiple reset put option. The contribution of the paper is twofold. Firstly, we study the single reset put option and reduce it to a free-boundary problem, and then we tackle the latter using mostly probabilistic arguments such as Markov property and local time-space calculus (see \cite{Pe-1}). We provide the probabilistic proof of the monotonicity
of optimal reset boundary in some situations. 
 We then obtain the same reset premium representation for the option price
and recursive integral equation for optimal reset boundary as in \cite{Kwok-1}. Secondly and most importantly,
we complement the work \cite{Kwok-2} and study the multiple reset option using local time-space formula and by induction in number $n$ of reset rights. To the best of our knowledge,
this is the first paper where the reset premium representations for the option prices with multiple rights and the recursive integral equations for each optimal reset boundaries were obtained. 
This result extends the results for the single reset option. Thirdly, we describe the numerical algorithm for pricing multiple reset options and plotting optimal reset boundaries. Fourthly,  the paper also recalls the put-call reset-shout parity, i.e., the reset put option problem is equivalent to shout call one in the sense that their optimal strategies coincide and their prices have a simple relationship. The same fact is true for reset call and shout put options, which was also observed in \cite{Kwok-2} and \cite{Goard}.

Finally, we also note that these contracts can be attractive for investors in cryptocurrency markets. The high volatility of cryptocurrencies calls for more flexible option contracts, which could give new opportunities to investors even if they see sudden jumps in the price of the underlying cryptocurrency.  It then seems natural to consider reset strike or shout options in this market as an alternative to plain vanilla call/put options.

The paper is organised as follows. 
In Section 2 we formulate the multiple reset put option problem as a multiple optimal stopping problem, which we reduce to the sequence of single optimal stopping problems.
In Section 3 we provide a put-call parity between reset and shout options.
Section 4 is devoted to the single reset put option. In Section 5 we derive
reset premium representation for the price of the multiple reset option and characterize the optimal reset boundaries
as the unique solution to nonlinear integral equations. Finally, Section 6 provides numerical algorithms for numerical solutions to the multiple reset option problem.

\vs{6pt}

\section{Formulation of the problem}

Let us suppose the probability space $(\Omega,\cal{F},\QQ)$, where $\QQ$ is the risk-neutral measure. In this paper, we study the multiple reset put option problem. The underlying asset price $X$ follows geometric Brownian motion 
\begin{equation} \label{SDE} \hs{6pc}
 dX_{t}=(r\m \delta)X_{t}\,dt+\sigma X_{t}\,dB_{t}\;\;\; (X_0=x)
 \end{equation}
  where $B$ is a $\QQ$-standard Brownian motion started at zero, $r>0$ is the interest rate, $\delta\ge 0$ is the dividend yield and $\sigma>0$ is the constant volatility coefficient. The solution $X$ to the stochastic differential equation \eqref{SDE} is given by
\begin{equation} \label{model} \hs{6pc}
 X_t=x \exp\Big((r\m\delta\m\sigma^2/2)t+\sigma B_t\Big)
 \end{equation}
 for $t\ge 0$, $x>0$.

By definition of the multiple reset put option with $n$ rights the payoff at maturity time $T$ is the following: if the buyer ``resets'' at the sequence of endogenous times $\tau_1<....\le \tau_n\le T$ (i.e. at least $\tau_1<T$, others may equal $T$), he obtains
 $(X_{\tau_1}\vee...\vee X_{\tau_n}-X_T)^+$ and if the buyer does not reset until $T$ (i.e. $\tau_1=T$), his payoff equals $(K-X_T)^+$, where $K>0$ is the original strike price. Hence the multiple reset put option allows to obtain sequentially at-the-money put options by reseting the strike at times $\tau_i$, $i=1,...,n$. Clearly, one should reset her/his $i$-th right only when $X_t >K_{i-1}$ with $K_i=X_{\tau_i}$, $i=1,...,n-1$ and $K_0=K$. This option of an European type since the payoff is delivered at $T$, however it has an ``American''-style feature of an early action - reseting.

 Now let us turn to the valuation of multiple reset put option. If the holder resets (or does not reset and $\tau_1=T$) at the sequence of stopping times $\tau:=(\tau_1,...,\tau_n)$ with respect to the natural filtration of $X$, then the payoff at maturity $T$  is given by
 \begin{equation} \label{payoff-n} \hs{6pc}
(X_{\tau_1}\vee...\vee X_{\tau_n}\vee K\m X_T)^+.
 \end{equation}
Thus the arbitrage-free price of the multiple reset put option at time $t\in[0,T)$ as a value function of the following optimal stopping problem
\begin{equation} \label{problem-n} \hs{6pc}
V^{RP}_n(t,x;K)=e^{-r(T-t)}\sup \limits_{t\le\tau\le T} \EE_{t,x}\left[(X_{\tau_1}\vee...\vee X_{\tau_n}\vee K\m X_T)^+\right]
 \end{equation}
where the supremum is taken over the set of sequences of stopping times $\{\tau=(\tau_1,...,\tau_n): 0\le \tau_1\le...\le\tau_n\le T\}$, the expectation $\EE_{t,x}$ is taken under $\QQ$ and $X_t=x$, and we include the discount factor $e^{-rT}$ since the payoff is delivered at $T$. We also emphasize the dependence of
$V^{RP}_n$ on the current strike level $K$.
It is important to note that since the terminal payoff function \eqref{payoff-n}  is not known at reset times $\tau_i$, the multiple optimal stopping problem \eqref{problem-n} falls into the class of optimal prediction problems. Hence we  need to reduce this problem to a standard optimal stopping problem with an adapted payoff function and also to reformulate it as the sequence of single stopping problems.
For this, let us now rewrite the payoff in the following way
 \begin{align} \label{payoff-n-2} \hs{0pc}
\max(X_{\tau_1}\vee...\vee X_{\tau_n}\vee K\m X_T,0)=\max(X_{\tau_2}\vee...\vee X_{\tau_n}\vee (X_{\tau_1}\vee K)\m X_T,0).
 \end{align}
 The equality \eqref{payoff-n-2} shows that the payoff  of the option with remaining $n$ rights and strike $K$ is the same as the payoff of the option with remaining $n-1$ rights left and the strike $X_{\tau_1}\vee K$.  Now by conditioning on $X_{\tau_1}$ we reduce \eqref{problem-n} to
 \begin{equation} \label{problem-n-2} \hs{6pc}
V^{RP}_n(t,x;K)=\sup \limits_{t\le\tau_1\le T} \EE_{t,x} \left[e^{-r(\tau_1-t)} V^{RP}_{n-1}(\tau_1,X_{\tau_1};X_{\tau_1}\vee K)\right].
 \end{equation}
 Therefore to solve the problem \eqref{problem-n} we need to tackle $n$ single optimal stopping problems sequentially starting from $n=1$ and $V^{RP}_0$ being the European put option price.
 This reduction is very natural as it says that the price of the reset option with $n$ rights can be obtained as the expected discounted value at the first optimal  reset time $\tau_1^*$ of the reset option with $n-1$ remaining  rights. The reduction \eqref{problem-n-2} was also stated in \cite{Kwok-2}. The reduction of multiple optimal stopping problems to the sequence of single optimal stopping problems is quite standard in the literature of swing options, see e.g. \cite{CT} and the recent work on finite horizon \cite{DeKit}. 
 The pricing of the swing options on finite horizon are different from reset/shout options due to the presence of so-called refracting period, which is the minimal period between two exercise times. 

 Let us now write the problem as
 \begin{equation} \label{problem-n-3} \hs{6pc}
V^{RP}_n(t,x;K)=\sup \limits_{t\le\tau\le T}\EE_{t,x}\left[ e^{-r(\tau-t)} G_n(\tau,X_{\tau})\right]
 \end{equation}
 for $n\ge 1$, $0\le t\le T$ and $x,K>0$ where the payoff function $G_n$ reads inductively
 \begin{equation} \label{payoff-n-3} \hs{6pc}
 G_n(t,x):=V^{RP}_{n-1}(t,x;x\vee K)
 \end{equation}
and
 \begin{equation} \label{problem-0} \hs{6pc}
V^{RP}_{0}(t,x;K)=V^e(t,x;x\vee K)=\EE_{t,x} \left[e^{-r(T-t)}(x\vee K\m X_{T})^+\right]
 \end{equation}
 is the price of European put option at time $t\in[0,T)$ and asset price $x>0$, with strike $x\vee K$ and expiry date $T$.
\vs{6pt}

\section{Parity between shout and reset options}

In this section we recall definitions of reset call option, shout put and call options with multiple rights, and show the parity between them. 

\begin{definition} Reset call option has the following payoff at maturity time $T$: if the buyer `resets' at sequence of times $\tau_1<....\le \tau_n<T$ he gets $(X_T -X_{\tau_1}\wedge...\wedge X_{\tau_n})^+$ and if the buyer does not `reset' until $T$ his payoff equals $(X_T -K)^+$, where $K>0$ is the original strike price.
\end{definition}

\begin{definition} Shout call option has the following payoff at maturity time $T$: if the buyer `shouts' at sequence of times $\tau_1<....\le \tau_n<T$ he gets $\max(X_{\tau_1}-K,..., X_{\tau_n}-K)$ and if the buyer does not `shout' until $T$ his payoff equals $(X_T -K)^+$, where $K>0$ is the strike price.
\end{definition}

\begin{definition} Shout put option has the following payoff at maturity time $T$:  if the buyer `shouts' at sequence of times $\tau_1<....\le \tau_n<T$ he gets $\max(K-X_{\tau_1},..., K-X_{\tau_n})$ and if the buyer does not `shout' until $T$ his payoff equals $(K-X_T )^+$, where $K>0$ is the strike price.
\end{definition}

  If the holder resets (or shouts) at sequence of times $\tau_1\le....\le \tau_n<T$, then the payoff for reset call option, and shout call and put options with $n$ rights at maturity $T$ under risk-neutral measure $\QQ$ is given by, respectively
 \begin{align}
 \label{resetcall}\hs{6pc}
&(X_T -X_{\tau_1}\wedge\ldots\wedge X_{\tau_n}\wedge K)^+\\
 \label{shoutcall}
 &\max(X_{\tau_1}-K,\ldots, X_{\tau_n}-K,X_T-K,0)\\
\label{shoutput}
&\max(K-X_{\tau_1},\ldots, K-X_{\tau_n},K-X_T,0)
 \end{align}
and their arbitrage-free prices
\begin{align}
\label{resetcallprice} \hs{6pc}
&V^{RC}_n=e^{-rT}\sup \limits_{0\le\tau\le T}\EE\left[(X_T -X_{\tau_1}\wedge\ldots\wedge X_{\tau_n}\wedge K)^+\right]\\
 \label{shoutcallprice}
&V^{SC}_n=e^{-rT}\sup \limits_{0\le\tau\le T}  \EE\left[\max(X_{\tau_1}-K,\ldots, X_{\tau_n}-K,X_T-K,0)\right]\\
\label{shoutputprice}
&V^{SP}_n=e^{-rT}\sup \limits_{0\le\tau\le T}\EE\left[\max(K-X_{\tau_1},\ldots, K-X_{\tau_n},K-X_T,0)\right].
 \end{align}
Straightforward calculations show the parity between the prices of reset call and shout put options
\begin{align} \label{parity-1} \hs{6pc}
V^{RC}_n&=e^{-rT}\sup \limits_{0\le\tau\le T}\EE\left[\max(X_T -X_{\tau_1},\ldots,X_T-X_{\tau_n},X_T- K,0)\right]\\
&=e^{-rT}\sup \limits_{0\le\tau\le T}\EE\left[X_T-K+\max(K -X_{\tau_1},\ldots,K-X_{\tau_n},K-X_T,0)\right]\nonumber\\
&=X_0 e^{-\delta T}-K e^{-rT}+V^{SP}_n\nonumber
 \end{align}
and between the prices of reset put and shout call options
\begin{align} \label{parity-2} \hs{6pc}
V^{RP}_n &=e^{-rT}\sup \limits_{0\le\tau\le T}\EE\left[\max(X_{\tau_1}-X_T,\ldots,X_{\tau_n}-X_T,K-X_T,0)\right]\\
&=e^{-rT}\sup \limits_{0\le\tau\le T}\EE\left[K-X_T+\max(X_{\tau_1}-K,\ldots,X_{\tau_n}-K,X_T-K,0)\right]\nonumber\\
&=K e^{-rT}-X_0 e^{-\delta T}+V^{SC}_n\nonumber
 \end{align}
and the relationship between all of them
\begin{align} \label{parity-3} \hs{6pc}
V^{RP}_n -V^{SC}_n=V^{SP}_n-V^{RC}_n=K e^{-rT}-X_0 e^{-\delta T}.
 \end{align}
 Finally, we note that the optimal stopping times are the same for reset call and shout put as well as for reset put and shout call options.
\vs{6pt}

 \section{Single reset put option}

 Throughout the rest of the paper we focus on the reset put option and for notational convenience we use $V_n$ instead of $V_n^{RP}$. In this section we study the pricing of single reset put option, i.e., $n=1$ in \eqref{problem-n-3}
  \begin{equation} \label{problem-1-2} \hs{6pc}
V_1(t,x;K)=\sup \limits_{t\le\tau\le T}\EE_{t,x} \left[e^{-r(\tau-t)} G_1(\tau,X_{\tau})\right]
 \end{equation}
 for $t\in[0,T)$ and $x,K>0$, and where   \begin{equation} \label{payoff-1-2} \hs{6pc}
 G_1(t,x)=\EE_{t,x}\left[ e^{-r(T-t)}(x\vee K-X_{T})^+\right].
  \end{equation}
  We solve the problem  \label{problem-1} for fixed $K$ so
 we will use notation  $V_1(t,x)$ instead of $V_1(t,x;K)$ in this section.

  We will reduce the problem  \label{problem-1} into a free-boundary problem and the latter will be tackled in the next section using local time-space calculus (\cite{Pe-1}).
 First using that the gain function $G_1$ is continuous and standard arguments (see e.g. Corollary 2.9 (Finite horizon) with Remark 2.10 in \cite{PS}) we have that continuation and reset sets read
\begin{align} \label{C} \hs{5pc}
&C_1= \{\, (t,x)\in[0,T)\! \times\! [0,\infty):V_1(t,x)>G_1(t,x)\, \} \\[3pt]
 \label{D}&D_1= \{\, (t,x)\in[0,T)\! \times\! [0,\infty):V_1(t,x)=G_1(t,x)\, \}
 \end{align}
and the optimal stopping time in \eqref{problem-1-2} is given by
\begin{align} \label{OST} \hs{5pc}
\tau_1^*=\inf\ \{\ t\leq s\leq T:(s ,X_{s})\in D_1\ \}.
 \end{align}

1. We will now discuss some properties of the functions $G_1$ and $V_1$. The gain function $G_1$ can be simplified
\begin{align} \label{G-1} \hs{2pc}
G_1(t,x)&=xV^e (t,1;1)=:xV^e(t)\\
&=x \left[e^{-r(T-t)} \Phi\Big((\tfrac{\delta-r}{\sigma}\p\tfrac{\sigma }{2})\sqrt{T-t}\Big)- e^{-\delta(T-t)}\Phi\Big((\tfrac{\delta-r}{\sigma}\m\tfrac{\sigma} {2})\sqrt{T-t}\Big)\right]\nonumber
\end{align}
 for $x>K$, where $\Phi(\cdot)$ is the cumulative distribution function of the standard normal law and $V^e(t)$ is the price of European  put option at $t$ with $x=1$ and strike 1. On the other hand, $G_1$ equals to
 \begin{align} \label{G-2} \hs{3pc}
G_1(t,x)=\;&K e^{-r(T-t)} \Phi\Big(\tfrac{1}{\sigma\sqrt{T-t}}\big[\log(\tfrac{K}{x})\m(r\m\delta\m\tfrac{\sigma^2 }{2})(T\m t)\big]\Big)\\
&-x e^{-\delta(T-t)}\;\Phi\Big(\tfrac{1}{\sigma\sqrt{T-t}}\big[\log(\tfrac{K}{x})\m(r\m\delta\p\tfrac{\sigma^2} {2})(T\m t)\big]\Big)\nonumber
\end{align}
for $0<x\le K$ and this is exactly the European put option price $V^e(t,x;K)$ with strike $K$. We know that $x\mapsto V^e(t,x)$ is convex on $(0,\infty)$ for any $0\le t<T$ fixed so that $x\mapsto G_1(t,x)$ is convex on $(0,K)$ for every $t$ fixed. Since $G_1$ is linear in $x$ on $(K,\infty)$ in order to prove that $x\mapsto G_1(t,x)$ is convex on $(0,\infty)$ for every $t$ fixed it is enough to show  that  $(G_1)_x (t,K+)\ge (G_1)_x (t,K-)$. Using \eqref{G-1}, \eqref{G-2} and well-known expression for `delta' coefficient of the European put option $\Delta_p=\frac{\partial V^e}{\partial x}
=-\Phi\Big(\tfrac{1}{\sigma\sqrt{T-t}}\big[\log(\tfrac{K}{x})\m(r\m\delta\p\tfrac{\sigma^2 }{2})(T\m t)\big]\Big)$ we have that
\begin{align} \label{G-3} \hs{3pc}
(G_1)_x (t,K+)&=e^{-r(T-t)} \Phi\Big((\tfrac{\delta-r}{\sigma}\p\tfrac{\sigma }{2})\sqrt{T-t}\Big)- \Phi\Big((\tfrac{\delta-r}{\sigma}\m\tfrac{\sigma} {2})\sqrt{T-t}\Big)\\
\label{G-4}
(G_1)_x (t,K-)&=-\Phi\Big((\tfrac{\delta-r}{\sigma}\m\tfrac{\sigma}{2})\sqrt{T-t}\Big)
\end{align}
and thus it is clear that $(G_1)_x (t,K+)\ge (G_1)_x (t,K-)$.
Hence the function $x\mapsto G_1(t,x)$ is convex on $(0,\infty)$ for every $t$ fixed and thus using standard arguments we obtain that $x\mapsto V_1(t,x)$ is convex on $(0,\infty)$ as well.

\vs{6pt}

2. It will be useful for further analysis to calculate the expression $H_1:=\frac{\partial G_1}{\partial t} \p \L_X G_1\m rG_1$ for $(t,x)\in[0,T)\times(0,\infty)$ where
$\L_X =(r\m\delta)xd/dx+(\sigma^2/2) x^2 d^2/dx^2$ is the infinitesimal operator of $X$. Since we have that $G_1(t,x)=V^e(t,x)$ for $(t,x)\in[0,T)\times(0,K)$ and it is well-known that $V^e_t\p \L_X V^e\m rV^e=0$ for all $(t,x)\in[0,T)\times(0,\infty)$ then
\begin{align} \label{H-1} \hs{3pc}
H_1(t,x)=0 \quad \text{on}\quad [0,T)\times(0,K).
\end{align}
Now we consider set $\{(t,x)\in[0,T)\times(K,\infty)\}$ and there $G_1$ is given by \eqref{G-1} so that  we have
\begin{align} \label{H-4} \hs{-0.5pc}
H_1(t,x)=x\left[(r\m \delta) e^{-r(T-t)}\Phi\big((\tfrac{\delta-r}{\sigma}\p\tfrac{\sigma }{2})\sqrt{T\m t}\big)-\frac{\sigma}{2\sqrt{T-t}} e^{-r(T-t)}\varphi\big((\tfrac{\delta-r}{\sigma}\p\tfrac{\sigma }{2})\sqrt{T\m t}\big)\right]
\end{align}
for $(t,x)\in[0,T)\times(K,\infty)$ and where $\varphi(\cdot)$ is the probability density function of the standard normal law.
 For convenience we define
\begin{align} \label{H-5} \hs{1pc}
h_1(t)=& (V^e(t))'\m \delta V^e(t)\\
=&(r\m \delta) e^{-r(T-t)}\Phi\big((\tfrac{\delta-r}{\sigma}\p\tfrac{\sigma }{2})\sqrt{T\m t}\big)-\frac{\sigma}{2\sqrt{T-t}} e^{-r(T-t)}\varphi\big((\tfrac{\delta-r}{\sigma}\p\tfrac{\sigma }{2})\sqrt{T\m t}\big)\nonumber
\end{align}
for $t\in[0,T)$ so that $H_1(t,x)=x h_1(t)$ for $(t,x)\in[0,T)\times(K,\infty)$.
\vs{2pt}

By straightforward calculations it can be seen from \eqref{H-5} that $h_1(t)\downarrow -\infty$ as $t\rightarrow T$ and $h_1(0)>0$ as $T\rightarrow\infty$. Moreover,
there exists large enough $T>0$ and $t_*<T$ such that $h_1$ is increasing on $[0,t_*)$ and is decreasing on $(t_*,T)$. Then for this size of $T$
there is a point $t^*>t_*$ with $h_1>0$ on $[0,t^*)$ and $h_1\le 0$ on $[t^*,T)$. We note that for typical values of parameters $(T, r, \delta, \sigma)$ we have that $h_1$ is decreasing and negative on $[0,T)$ and thus positive points $t_*$ and $t^*$ do not exist.
\vs{6pt}

3. Below we derive some insights into the structure of reset set $D_1$.
Using Ito-Tanaka's formula and \eqref{H-4} with \eqref{G-3}-\eqref{G-4} we have
\begin{align} \label{Tanaka-1} \hs{1pc}
 \EE_{t,x}\left[ e^{-r(\tau-t)}G_1(\tau,X_\tau)\right]=\;&G_1(t,x)+\EE_{t,x}\left[ \int_t^\tau e^{-r(s-t)}h_1(s)X_s I(X_s \ge K)ds\right]\\
 &+\frac{1}{2} \EE_{t,x}\left[ \int_t^\tau e^{-r(s-t)} \Phi\Big(-\left(\tfrac{r}{\sigma}\m\tfrac{\sigma }{2}\right)\sqrt{T\m  s}\Big)d\ell^{K}_s (X)\right]\nonumber
 \end{align}
for $(t,x)\in[0,T)\times(0,\infty)$ and any stopping time $\tau$ of $X$, where $(\ell^K_s(X))_{s\ge t}$ is the local time process of $X$ at level $K$.
The formula \eqref{Tanaka-1} shows that for $x<K$ the first integrand on the right-hand side vanishes and the buyer should continue due to
the non-negative integral with respect to local time at $K$. When $x>K$ and the integrand $h_1$ is
positive, i.e. $t<t^*$, then it is again optimal to wait. Now if
$x>K$ and $h_1(t)\le 0$, i.e. $t\ge t^*$, the only incentive to wait is to accumulate a positive local time at $K$ in future up to $T$.
However this makes sense only for asset prices above but close to $K$ and for large $x$ one should reset the strike immediately due to very negative $ h_1 X$.
To conclude, we obtain that it is optimal to reset only when $t\ge t^*$ (if $t^*$ exists as described above) and large enough $x$ above $K$.
\vs{2pt}

Now we can describe the structure of the reset set $D_1$.
Namely, from the fact that it is not optimal to exercise the reset put option below $K$, definitions \eqref{C} and \eqref{D}, convexity of $V_1$ and linearity of $G_1$ above $K$ it follows that
there exists an optimal reset boundary $b_1:[0,T]\rightarrow \R$ such that
\begin{align} \label{OST-2} \hs{5pc}
&\tau_{1}^*=\inf\ \{\ t\leq s\leq T:X_{s}\ge b_1(s) \}
 \end{align}
is optimal in \eqref{problem-1-2} and $b_1(t)>K$ for $t\in[0,T)$. Also $b_1\equiv =+\infty$ on $[0,t^*)$  if $t^*$ exists. It is clear that if $x>K$ and $t<T$ is sufficiently close to $T$ then it is optimal to reset
immediately (since $h_1<0$ near $T$ and the profit obtained from the local time term at $K$ cannot compensate the cost of getting
there due to the lack of remaining time). This shows that $b_1(T-)=K$.

\vs{6pt}

4. Standard Markovian arguments (see \cite{PS}) lead to the following free-boundary problem (for the value
function $V_1=V_1(t,x)$ and the optimal reset boundary $b_1=b_1(t)$ to be determined):
\begin{align} \label{PDE} \hs{5pc}
&\frac{\partial V_1}{\partial t} \p\L_X V_1\m rV_1=0 &\hs{-30pt}\text{in}\;  C_1\\
\label{IS}&V_1(t,b_1(t))=G_1(t,b_1(t)) &\hs{-30pt}\text{for}\; t\in[0,T)\\
\label{SF}&\frac{\partial V_1}{\partial x} (t,b_1(t))=\frac{\partial G_1}{\partial x}(t,b_1(t)) &\hs{-30pt}\text{for}\; t\in[0,T) \\
\label{TC}&V_1(T,x)=(K\m x)^+ &\hs{-30pt}\text{for}\; x\in(0,\infty)\\
\label{FBP1}&V_1(t,x)>G_1(t,x) &\hs{-30pt}\text{in}\; C_1\\
\label{FBP2}&V_1(t,x)=G_1(t,x) &\hs{-30pt}\text{in}\; D_1
\end{align}
where the continuation set $C_1$ and the reset set $D_1$ are given by
\begin{align} \label{C-1} \hs{5pc}
&C_1= \{\, (t,x)\in[0,T)\! \times\! (0,\infty):x<b_1(t)\, \} \\[3pt]
 \label{D-1}&D_1= \{\, (t,x)\in[0,T)\! \times\! (0,\infty):x\ge b_1(t)\, \}
 \end{align}
with  $b_1=+\infty$ on $[0,t^*)$ if $t^*>0$ exists. It can be shown that this free-boundary problem has a unique solution $V_1$ and $b_1$ which coincide with the
value function and the optimal reset boundary, respectively.
\vs{2pt}

Completed details of the analysis above go beyond our goals in this paper and for this reason
will be omitted. It should be noted however that one of the main issues which makes this
analysis quite complicated (in comparison with e.g. the American put option problem)
is that the proof of monotonicity in time of $b_1$ is still an open question for certain values of parameters. This property has been stated in previous papers on reset and shout options
but actually have not been proven apart from \cite{Dai-1}, where PDE arguments were used to prove the monotonicity for some occasions. We provide the  probabilistic proof for the same situations, i.e., when $h_1$ is decreasing on $[0,T)$ and
 $r\le \sigma^2/2$. In this case, we can easily see from \eqref{Tanaka-1} that both integral terms are decreasing in $t$ and obtain that $\partial V_1/\partial t\le \partial G_1/\partial t$
 and therefore $b_1$ is decreasing. In the opposite case, when $h_1$ be increasing on $[0,t_*)$ if such $t_*>0$ exists or $r>\sigma^2/2$,
 both integrals on  \eqref{Tanaka-1} are not decreasing at the same time and it is not clear how to compare them.
 However, numerical results clearly support the hypothesis that $b_1$ is decreasing, and it was also obtained numerically in other papers too.
\vs{2pt}

Another difficulty of the analysis is that the proof of the continuity of the free boundary without having its
monotonicity is a challenging problem, which can help to tackle other optimal stopping problems.
In the next paragraph we will derive simpler equations
which characterize $V_1$ and $b_1$ uniquely and can be used for the pricing of the reset put option and computation of the optimal reset boundary.

\vs{6pt}

5. We now provide the reset premium representation formula for the arbitrage-free price $V_1$ which decomposes it into the sum of the European put option price and the `reset' premium. The optimal reset boundary $b_1$ will be obtained as the unique solution to the integral equation.

We will make use of the following function in Theorem \ref{th:1} below:
\begin{align} \label{L} \hs{1pc}
L_1(t,u,x,z)&=- e^{-r(u-t)}\EE_{t,x} \left[h_1(u)X_{u} I(X_{u} \ge z)\right]\\
&=-xe^{-\delta(u-t)} h_1(u)\wh{\QQ} (X_{u} \ge z)\nonumber\\
&=-xe^{-\delta(u-t)} h_1(u)\Phi\left[\left(\log(x/z)+(r\m\delta\p\sigma^2/2)(u-t)\right)/(\sigma\sqrt{u-t})\right]\nonumber
 \end{align}
 for $t,u\ge 0$ and $x,z>0$ and where the measure $\wh{\QQ}$ is defined as $d\wh{\QQ}=e^{\sigma B_T-\sigma^2 T/2}\;d\QQ$ and we used that $\wh{B}_t=B_t+\sigma t$ is a standard Brownian motion under
 $\wh{\QQ}$.
 \vs{6pt}

 The main result of this section may now be stated as follows.
\begin{theorem}\label{th:1}
 The optimal reset boundary $b_1$ can be characterised as the unique solution to the nonlinear integral equation
\begin{align}\label{th-2} \hs{3pc}
b_1(t) V^e(t)=\;V^e (t,b_1(t)) +\int_t^{T} L_1(t,u,b_1(t),b_1(u))du
\end{align}
for $t\in[0,T]$ in the class of continuous functions above $K$ with $b_1(T-)=K$. Then the value function $V_1$ in \eqref{problem-1-2} has the following representation
\begin{align}\label{th-1} \hs{3pc}
V_1(t,x)=\;V^e (t,x) +\int_t^{T}L_1(t,u,x,b_1(u))du
\end{align}
for $t\in[0,T]$ and $x\in (0,\infty)$, where $V_1$ is decomposed as the sum of the European put price $V^e$ and the reset premium.
\end{theorem}
\begin{proof}

$(A)$ First we clearly have that the following conditions hold:
$(i)$ $V_1$ is $C^{1,2}$ on $C_1\cup D_1$;
$(ii)$ $b_1$ is of bounded variation;
 $(iii)$ $(V_1)_t+\L_{X} V_1- rV_1$ is locally bounded;
 $(iv)$ $x\mapsto V_1(t,x)$ is convex (recall paragraph 1 above);
$(v)$ $t\mapsto (V_1)_x (t,b_1(t)\pm)$ is continuous (recall \eqref{SF}). Hence we can apply the local time-space formula on curves \cite{Pe-1} for $e^{-r(s-t)}V_1( s,X_s)$:
\begin{align} \label{n-proof-1} \hs{1pc}
e^{-r(s-t)}V_1(&s,X_s)\\
=\;&V_1(t,x)+M_s\nonumber\\
 &+ \int_t^{s}
e^{-r(u-t)}\left(\tfrac{\partial V_1}{\partial t}  \p\L_X V_1\m rV_1\right)(u,X_u)
 I(X_u > b_1(u))du\nonumber\\
 &+\frac{1}{2}\int_t^{s}
e^{-r(u-t)}\left(\tfrac{\partial V_1}{\partial x}  (u,X_u +)-\tfrac{\partial V_1}{\partial x}  (u,X_u -)\right)I\big(X_u=b_1( u)\big)d\ell^{b_1}_u(X)\nonumber\\
  =\;&V_1(t,x)+M_s + \int_t^{s}e^{-r(u-t)}
\left(\tfrac{\partial G_1}{\partial t}  \p\L_X G_1\m G_1)( u,X_u\right)I(X_u \ge b_1(u))du\nonumber\\
  =\;&V_1(t,x)+M_s +\int_t^{s}e^{-r(u-t)}h_1(u) X_u I(X_u \ge b_1(u))du\nonumber
  \end{align}
where we used \eqref{PDE}, \eqref{H-4}, the smooth-fit condition \eqref{SF} and where $M=(M_u)_{u\ge 0}$ is the martingale part,  $(\ell^{b_1}_u(X))_{u\ge 0}$ is the local time process of $X$ spending at the boundary $b_1$. Now
upon letting $s=T$, taking the expectation $\EE_{t,x}$,  the optional sampling theorem for $M$, rearranging terms and noting that $V_1(T,x)=G_1(T,x)=(x\vee K\m x)^+=(K\m x)^+$ for all $x>0$, we get \eqref{th-1}.
The integral equation \eqref{th-2} one obtains by simply putting $x=b_1(t)$ into \eqref{th-1} and using \eqref{IS}.

\vs{6pt}

$(B)$ The proof of that $b_1$ is the unique solution to the equation \eqref{th-2} in the class of continuous functions $t\mapsto b_1(t)$
 is based on standard arguments originally derived in \cite{Pe-4} and omitted here.

\end{proof}
\vs{6pt}

 \section{Multiple reset put option}

 In this section we turn to the multiple reset put option problem
 \begin{equation} \label{problem-n-4} \hs{5pc}
V_n(t,x)=V_n(t,x;K)=\sup \limits_{t\le\tau\le T}\EE\left[ e^{-r(\tau-t)} G_n(\tau,X_{\tau})\right]
\end{equation}
 for $n\ge 1$, $0\le t\le T$ and $x,K>0$ where the gain function $G_n$ is given inductively
 \begin{equation} \label{payoff-n-4} \hs{5pc}
 G_n(t,x):=V_{n-1}(t,x;x\vee K)
 \end{equation}
 for $0\le t\le T$ and $x>0$. 
Having solved \eqref{problem-1-2} for $n=1$ in the previous section, we can derive the pricing formulas and integral equations for general $n\ge2$ by induction.

\begin{theorem}\label{th:n.1}
For any $n\ge 1$ the optimal stopping time in  \eqref{problem-n-4} is given by
\begin{align} \label{n-OST-2} \hs{5pc}
\tau^*_{n}=\inf\ \{\ t\leq s\leq T:X_{s}\ge b_n(s) \}
 \end{align}
 where the optimal reset boundary $b_n$ can be characterised as the unique solution to the nonlinear integral equation
\begin{align}\label{n-th-2} \hs{5pc}
b_n(t)(V^{e}(t)+p_{n-1}(t))=\;V^{e}(t,b(t)) +\int_t^T L_n(t,u,b_n(t),b_n(u))du
\end{align}
for $t\in[0,T]$ in the class of continuous functions above $K$ with $b_n(T-)=K$.
 The value function $V_n$ is characterized by the reset premium representation
\begin{align}\label{n-th-1} \hs{5pc}
V_n(t,x)=\;V^{e}(t,x) +\int_t^{T}L_n(t,u,x,b_n(u))du
\end{align}
for $t\in[0,T]$ and $x\in (0,\infty)$, and where
\begin{align}\label{n-th-3} \hs{3pc}
&L_n(t,u,x,z)=- x e^{-\delta(u-t)}h_n(u)\wh{\QQ} \left(X^x_{u-t} \ge z\right)\\
\label{n-th-4}&h_n(t)=(V^e (t)+p_{n-1}(t))'-\delta (V^e(t)+p_{n-1}(t))\\
\label{n-th-5}&p_n(t)=-\int_t^{T}e^{-\delta(u-t)}h_n(u)\wh{\QQ}_{t,K} \left(X_{u} \ge b_n(u;K)\right)du\\
\label{n-th-6}&p_0(t)\equiv 0.
\end{align}
Moreover, $V_n$ and $b_n$ solve the following free-boundary problem
\begin{align} \label{n-PDE} \hs{3pc}
&\frac{\partial V_n}{\partial t} \p\L_X V_n\m rV_n=0 &\hs{-30pt}\text{in}\;  C_n\\
\label{n-IS}&V_n(t,b_n(t))=G_n(t,b_n(t))=b_n(t) \left(V^e (t)+p_{n-1}(t)\right) &\hs{-30pt}\text{for}\; t\in[0,T)\\
\label{n-SF}&\frac{\partial V_n}{\partial x}(t,b_n(t))=\frac{\partial G_n}{\partial x}(t,b_n(t))=V^e (t)+p_{n-1}(t) &\hs{-30pt}\text{for}\; t\in[0,T) \\
\label{n-TC}&V_n(T,x)=(K\m x)^+ &\hs{-30pt}\text{for}\; x\in(0,\infty)\\
\label{n-FBP1}&V_n(t,x)>G_n(t,x) &\hs{-30pt}\text{in}\; C_n\\
\label{n-FBP2}&V_n(t,x)=G_n(t,x) &\hs{-30pt}\text{in}\; D_n
\end{align}
where the continuation set $C_n$ and the stopping set $D_n$ are given by
\begin{align} \label{C-1} \hs{3pc}
&C_n= \{\, (t,x)\in[T_0,T_n)\! \times\! (0,\infty):x<b_n(t)\, \} \\[3pt]
 \label{D-1}&D_n= \{\, (t,x)\in[T_0,T_n)\! \times\! (0,\infty):x\ge b_n(t)\, \}.
 \end{align}

\end{theorem}

\begin{proof}

1. The proof is followed by induction. In the previous section we verified the statement of theorem for $n=1,$ so now we assume that the formulas hold for $n-1$. In what follows we prove   \eqref{n-th-2} and \eqref{n-th-1} for $n$. The proof will follow along the same lines as for $n=1$. 

We have that the payoff function $G_n$
 is
\begin{align} \label{n-G-1} \hs{3pc}
G_n(t,x)=V_{n-1}(t,x;K)
\end{align}
for $0<x\le K$ and thus the local benefits of waiting
\begin{align} \label{n-H-1} \hs{3pc}
H_n(t,x):=\left(\frac{\partial G_n}{\partial t} \p \L_X G_n\m rG_n\right)(t,x)=0 \quad \text{on}\quad [0,T)\times(0,K)
\end{align}
where we used \eqref{n-PDE} for $n-1$. Now if $x>K$, using the assumption of induction we obtain
\begin{align} \label{n-G-2} \hs{3pc}
G_n(t,x)&=V_{n-1}(t,x;x)\\
&=V^{e}(t,x;x)\nonumber
+\int_t^{T}L_{n-1}(t,u,x,b_n(u;x))du\\
&=x \left(V^e (t)-\int_t^{T}e^{-\delta(u-t)} h_{n-1}(u)  \wh{\QQ} (X^x_{u-t} \ge b_{n-1}(u;x))du\right)\nonumber\\
&=x \left(V^e (t)+p_{n-1}(t)\right)\nonumber
\end{align}
so that
  \begin{align} \label{H-4} \hs{3pc}
H_n(t,x)=x\left(V^e (t)+p_{n-1}(t) \right)'-\delta x (V^e(t)+p_{n-1}(t))=x h_n(t)
\end{align}
for $(t,x)\in[0,T)\times(K,\infty)$.
\vs{6pt}

2. Now by applying Ito-Tanaka's formula for $e^{-r(s-t)}G_n(s,X_s)$ we have
\begin{align} \label{n-Tanaka-1} \hs{1pc}
 \EE_{t,x}&\left[ e^{-r(\tau-t)}G_n(\tau,X_\tau)\right]\\
 =&\;G_n(t,x)+\EE_{t,x}\left[ \int_t^\tau e^{-r(s-t)}H_n( s,X_s)I(X_s \ge K)ds\right]\nonumber\\
 &+\frac{1}{2} \EE_{t,x}\left[ \int_t^\tau e^{-r(s-t)}\left(\frac{\partial G_n}{\partial x} (s,K+)\m \frac{\partial G_n}{\partial x} ( s,K-)\right)d\ell^{K}_s (X)\right]\nonumber
 \end{align}
for $(t,x)\in[0,T)\times(0,\infty)$ and any stopping time $\tau$ of $X$.
 The formula \eqref{n-Tanaka-1} shows that for $X_t<K$ it is not optimal to reset at once.
\vs{2pt}

Using similar arguments as for $n=1$ we can deduce that
there exists an optimal reset boundary $b_n:[0,T]\rightarrow \R$ such that
\begin{align} \label{n-OST-2} \hs{5pc}
&\tau_{n}^*=\inf\ \{\ t\leq s\leq T:X_{s}\ge b_n(s) \}
 \end{align}
is optimal in \eqref{problem-n-4} and $b_n(t)>K$ for $t\in[0,T)$ with $b_n(T-)=K$.
\vs{6pt}

3. Using standard Markovian arguments one can show that $V_n$ and the optimal reset boundary $b_n$ uniquely solve the
free-boundary problem \eqref{n-PDE}-\eqref{n-FBP2}. As for the single reset option problem in previous section, the main challenge here
is to prove the monotonicity of the boundary $b_n$. We note that due to complicated expression for $h_n$, it is even more difficult to prove this property rather than for $n=1$, where we could show it for some cases. However numerical analysis in the next section and results
in \cite{Kwok-2} supports the claim that all boundaries are monotone. 
\vs{2pt}

We now apply the local time-space formula on curves \cite{Pe-1} for $e^{-rs}V_n(t\p s,X^x_s)$:
\begin{align} \label{n-proof-1} \hs{1pc}
e^{-r(s-t)}V_n(& s,X_s)\\
=\;&V_n(t,x)+M_s\nonumber\\
 &+ \int_t^{s}
e^{-r(u-t)}\left(\tfrac{\partial V_n}{\partial t}  \p\L_X V_n\m rV_n\right)(u,X_u)
 I(X_u \ge b_n(u))du\nonumber\\
 &+\frac{1}{2}\int_t^{s}
e^{-r(u-t)}\left(\tfrac{\partial V_n}{\partial x}  (u,X_u +)-\tfrac{\partial V_n}{\partial x}  (u,X_u -)\right)I\big(X_u=b_n( u)\big)d\ell^{b_n}_u(X)\nonumber\\
  =\;&V_n(t,x)+M_s + \int_t^{s}e^{-r(u-t)}
\left(\tfrac{\partial G_n}{\partial t}  \p\L_X G_n\m G_n)( u,X_u\right)I(X_u \ge b_n(u))du\nonumber\\
  =\;&V_n(t,x)+M_s +\int_t^{s}e^{-r(u-t)}h_n(u) X_u I(X_u \ge b_n(u))du\nonumber
  \end{align}
where we used \eqref{n-PDE}, \eqref{H-4}, the smooth-fit condition \eqref{n-SF} and where $M=(M_u)_{u\ge 0}$ is the martingale part,  $(\ell^{b_n}_u(X))_{u\ge 0}$ is the local time process of $X$ spending at the boundary $b_n$. Now
upon letting $s=T$, taking the expectation $\EE_{t,x}$,  the optional sampling theorem for $M$, rearranging terms, noting that $V_n(T,x)=(K\m x)^+$ for all $x>0$ and
recalling \eqref{n-th-3}, we finally get \eqref{n-th-1}.
 One obtains The integral equation \eqref{n-th-2} by simply putting $x=b_n(t)$ into \eqref{n-th-1} and using \eqref{n-IS}.
\end{proof}

\vs{6pt}

\section{Numerical algorithm and implementation}

In this section we provide a numerical algorithm to compute prices of reset put option and optimal reset boundaries. Also we
explain how to use these results to provide the sequential reset strategies for the buyer, and finally we plot numerically the multiple optimal reset boundaries and
the option prices.

\begin{figure}[t]
\begin{center}
\includegraphics[scale=1]{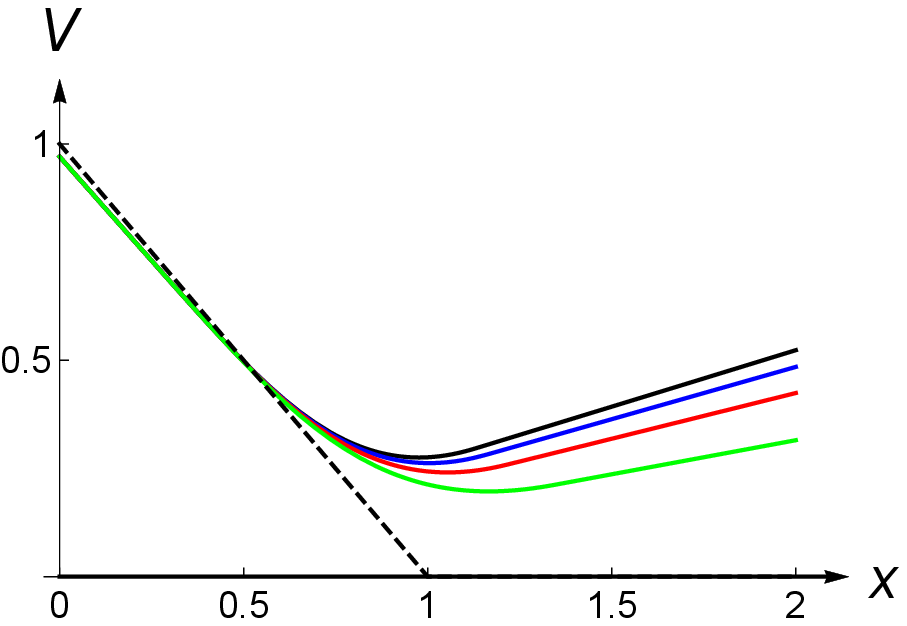}
\end{center}

{\par \leftskip=1.6cm \rightskip=1.6cm \small \ni \vs{-10pt}

\textbf{Figure 1.} This figure plots the prices of multiple reset put options \eqref{problem-n-4} for $n=1$ (green line), $n=2$ (red),
$n=3$ (blue) and $n=4$ (black). The parameter set is $T=1$, $K=1$, $r=0.03$, $\delta=0.04$, $\sigma=0.4$.

\par} \vs{10pt}

\end{figure}

\begin{figure}[t]
\begin{center}
\includegraphics[scale=1]{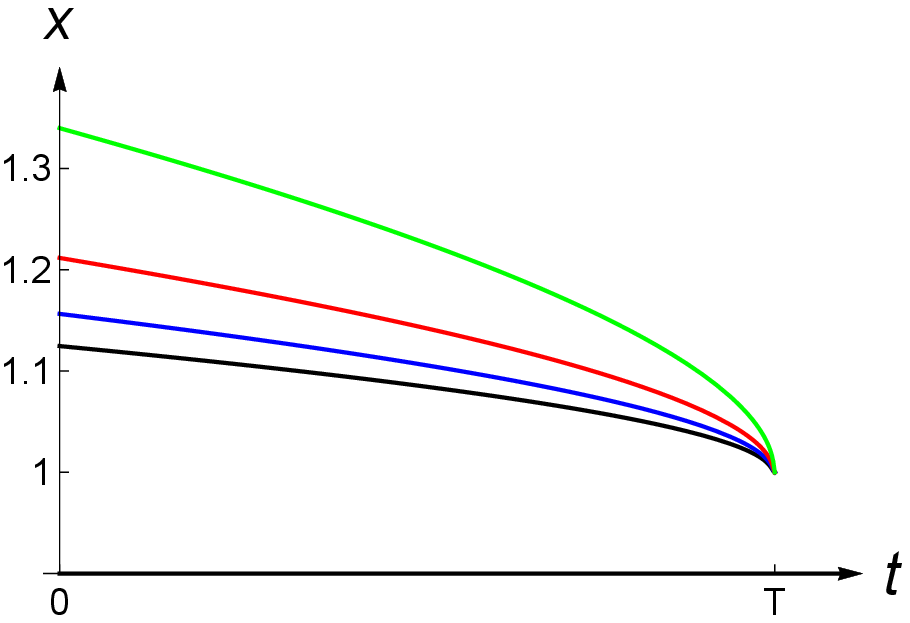}
\end{center}

{\par \leftskip=1.6cm \rightskip=1.6cm \small \ni \vs{-10pt}

\textbf{Figure 2.} This figure plots the optimal reset boundaries $b_n$ in \eqref{problem-n-4} for $n=1$ (green line), $n=2$ (red),
$n=3$ (blue) and $n=4$ (black). The parameter set is $T=1, K=1, r=0.03, \delta=0.04, \sigma=0.4$. 

\par} \vs{10pt}

\end{figure}

1. Now we describe the numerical solution to multiple reset put option problem \eqref{problem-n-4} and provide step-by-step algorithm.
\vs{2pt}

\emph{Step 1}. We calculate the function $h_1$ from \eqref{H-5}.
\vs{2pt}

\emph{Step 2}. Next, we solve numerically
the integral equation \eqref{th-2}.
Let us set $t_k=k\Delta$ for $k=0,1,...,N$ and $N$ large enough where $\Delta=T/N$ so that
the following discrete approximation of the integral equation  \eqref{th-2} is valid
\begin{align}\label{Alg-1} \hs{2pc}
G_1(t_k,b_1(t_k))=V^e (t_k,b_1(t_k)) +\Delta\sum_{l=k}^{N-1} L_1(t_k,t_{l+1},b_1(t_k),b_1(t_{l+1}))
\end{align}
for $k= 0,1,...,N\m1$. Setting $k=N\m1$ and $b_1(t_N)=K$ we can solve the algebraic equation
\eqref{Alg-1} numerically and get number $b_1(t_{N-1})$. Setting $k=N\m2$ and using the values $b_1(t_{N- 1})$,
$b_1(t_{N})$, we can solve \eqref{Alg-1} numerically and get number $b_1(t_{N- 2})$.
Continuing the recursion we obtain $b_1(t_{N-1}),\ldots,b_1(t_1),b_1(t_0)$  as approximations of the optimal boundary $b_1$ at the
points $T- \Delta,\ldots, \Delta, 0$.
We can then interpolate $b_1(\cdot)$ between points $(t_{k-1},t_k)$, $k=1,...,N$.
\vs{2pt}

\emph{Step 3}. The value function \eqref{th-1} can be approximated as follows
\begin{align}\label{Alg-2} \hs{2pc}
V_1(t_k,x)=V^e (t_k,x) +\Delta\sum_{l=k}^{N-1} L_1(t_k,t_{l+1},x,b_1(t_{l+1}))
\end{align}
for $k= 0,1,...,N\m 1$ and $x>0$, then interpolate $V_1(\cdot,x)$ between points $(t_{k-1},t_k)$, $k=1,...,N$.
\vs{2pt}

\emph{Step 4}. We approximate $p_1$ from \eqref{n-th-5} as
\begin{align}\label{Alg-2} \hs{5pc}
p_1(t_k)=-\Delta \sum_{l=k}^{N-1} e^{-\delta(t_{l+1}-t_k)} h_1(t_{l+1})\wh{\QQ}_{t_k,K} \left(X_{t_{l+1}} \ge b_1(t_{l+1})\right)
\end{align}
for $k= 0,1,...,N\m 1$ and interpolate between points $(t_{k-1},t_k)$, $k=1,...,N$.  We can then approximate $p'_1$.
\vs{2pt}

\emph{Step 5}. We compute the function $h_2$ using \eqref{n-th-4}.
\vs{2pt}

\emph{Step 6}. Now we solve \eqref{n-th-2} for $b_2$ using the same way as in \emph{Step 2} above. All components in the equation \eqref{n-th-2}
have been computed in previous steps.
\vs{2pt}

\emph{Step 7}. We calculate $V_2$ using \eqref{n-th-1}.
\vs{2pt}

\emph{Step 8}. We compute $p_2$ and $p'_2$ using \eqref{n-th-5} .
\vs{2pt}

\emph{Step 9}. We then continue the same procedure as above and compute $h_3 \rightarrow b_3 \rightarrow V_3 \rightarrow (p_3, p_3') \rightarrow \ldots
\rightarrow h_n \rightarrow b_n \rightarrow V_n \rightarrow (p_n, p_n')$.
\vs{6pt}

2. We note that the algorithm above is easy to implement, convergent as $N$ increases, and computationally quite fast.
In \cite{Kwok-1} the numerical comparison of recursive integral approach with binomial method has been provided for the single reset put option, and it was mentioned that for a given numerical accuracy, the former method is faster than the latter. This observation should be more apparent for multiple reset options studied in the current paper. 
\vs{2pt}

Figures 1 and 2 display the sequences of  the  option prices and  the optimal reset boundaries, respectively, for $n=4$. 
Figure 1 confirms that the reset option prices increase in the number of remaining reset rights $n$, which is an obvious fact. Also it shows $U$-shaped 
option prices for all $n\ge 1$ unlike the European put option price which is decreasing in $x$.
Figure 2 provides an evidence that for fixed $n$ the optimal reset boundaries are decreasing in $t$ and that for fixed $t\in[0,T)$ the boundaries 
are decreasing in $n$. 
\vs{2pt}

At last, we point out the following issue. When we obtain the sequence of the boundaries
$(b_k(\cdot))_{k=1}^n$ for the reset put option with $n$ rights, for the sake of unity we use the same strike level $K$. However,
using the optimal multiple reset policy the buyer will have different strike for each reset right, i.e.,
for $i$-th right the strike is $K_i=X_{\tau_{i-1}}$, $i=2,...,n$ and $K_1=K$. Therefore when the holder has $n$ remaining rights and resets at time $t$,
he immediately obtains  the reset put option with $n-1$ rights and the strike $X_t$, and thus solves the integral equation for next right with $K=X_t$.

\vs{12pt}

\begin{center}

\end{center}

\vs{16pt}


\begin{thebibliography}{99}
\bibitem{WFV-1} \textsc{Windcliff, H., Forsyth, P. A.  \emph{and} Vetzal, K. R.} (2001). Shout options: a framework for pricing contracts which can be modified by the investor. \emph{J. Comput. Appl. Math} 134 (213--241).
\end{thebibliography}
\end{document}